\documentclass[letterpaper, 10 pt, conference]{ieeeconf}  % Comment this line out

\usepackage{graphicx}
\usepackage{amsmath}
\usepackage{mathtools}
\usepackage{amssymb}
\usepackage{amsfonts}
\usepackage{upgreek}
\usepackage{hyperref}
\usepackage{cite}
\newcommand{\parallelogram}{%
    \tikz[baseline=0ex, line width = .7pt]{
        \draw (0,0) -- ++(0.3,0) -- ++(0.1,0.2) -- ++(-0.3,0) -- cycle;
    }%
}

\usepackage{algorithm}
\usepackage{algorithmic}
\usepackage{mathdots}
\usepackage{bm}
\usepackage{textcomp}
\usepackage{cases}
\usepackage[font=footnotesize]{caption}
\usepackage[font=footnotesize]{subcaption}

\usepackage{color}
\usepackage{soul}
\usepackage[normalem]{ulem}
\usepackage{graphbox}
\usepackage{wasysym}
\usepackage{xcolor}
\usepackage{mathtools}
\usepackage{svg}
\usepackage{float}
\usepackage{tikz}
\usepackage{pgfplots}
\pgfplotsset{compat=1.18}
\usetikzlibrary{intersections}
\usetikzlibrary{patterns}
\usepgfplotslibrary{fillbetween}
\usetikzlibrary{arrows.meta}
\usepackage{comment}
\newtheorem{lemma}{Lemma}
\newtheorem{remark}{Remark}
\newtheorem{theorem}{Theorem}
\newtheorem{assumption}{Assumption}

\def\@IEEEtablestring{table}

\long\def\@makecaption#1#2{%
	\ifx\@captype\@IEEEtablestring%
	\begin{center}{\footnotesize #1}\\{\footnotesize\scshape #2}\end{center}%
	\@IEEEtablecaptionsepspace% V1.6 was a hard coded 8pt
	\else
	\@IEEEfigurecaptionsepspace% V1.6 was a hard coded 5pt
	\setbox\@tempboxa\hbox{\footnotesize #1.~~ #2}%
	\ifdim \wd\@tempboxa >\hsize%
	\setbox\@tempboxa\hbox{\footnotesize #1.~~ }%
	\parbox[t]{\hsize}{\footnotesize \noindent\unhbox\@tempboxa#2}%
	\else%
	\ifcenterfigcaptions \hbox to\hsize{\footnotesize\hfil\box\@tempboxa\hfil}%
	\else \hbox to\hsize{\footnotesize\box\@tempboxa\hfil}%
	\fi\fi\fi}

\IEEEoverridecommandlockouts                              % This command is only
\title{\LARGE \bf
3D Directed Formation Control with Global Shape Convergence using Bispherical Coordinates \\
(extended version)}

\author{Omid Mirzaeedodangeh, Farhad Mehdifar, and Dimos V. Dimarogonas% <-this % stops a space
\thanks{This work is supported by ERC CoG LEAFHOUND, the KAW foundation, and the Swedish Research Council (VR).}% <-this % stops a space
\thanks{O. Mirzaeedodangeh, F. Mehdifar and D. V. Dimarogonas are with the Division of Decision and Control Systems, KTH Royal Institute of Technology, Stockholm, Sweden.   {\tt\small omidm@kth.se; mehdifar@kth.se; dimos@kth.se}}%%
}

\begin{document}

\maketitle
\thispagestyle{empty}
\pagestyle{empty}

%%%%%%%%%%%%%%%%%%%%%%%%%%%%%%%%%%%%%%%%%%%%%%%%%%%%%%%%%%%%%%%%%%%%%%%%%%%%%%%%
\begin{abstract}
In this paper, we present a novel 3D formation control scheme for directed graphs in a leader-follower multi-agent setup,, achieving (almost) global convergence to the desired shape. Specifically, we introduce three controlled variables representing bispherical coordinates that uniquely describe the formation in 3D. Acyclic triangulated directed graphs (a class of minimally acyclic persistent graphs) are used to model the inter-agent sensing topology, while the agents' dynamics are governed by single-integrator model.
Our analysis demonstrates that the proposed decentralized formation controller ensures (almost) global asymptotic stability 
while avoiding potential shape ambiguities in the final formation.
Furthermore, the control laws are implementable in arbitrarily oriented local coordinate frames of follower agents using only low-cost onboard vision sensors, making it suitable for practical applications. Finally, we validate our formation control approach by a simulation study.
\end{abstract}
%%%%%%%%%%%%%%%%%%%%%%%%%%%%%%%%%%%%%%%%%%%%%%%%%%%%%%%%%%%%%%%%%%%%%%%%%%%%%%%%
\section{Introduction}
Formation control in multiagent systems has undergone extensive investigation over the past decade.  Depending on the sensing and controlled variables, previous research efforts can be broadly classified into \cite{oh2015survey}, \cite{anderson2008rigid}: position-based \cite{ren2007information}, displacement-based \cite{ji2007distributed}, distance-based 
\cite{MEHDIFAR2020109086},
\cite{summers2009control}, bearing-based \cite{zhao2015bearing}, and angle-based  
\cite{{jing2019angle},{buckley2021infinitesimal}, 
{chen2023angle}, {chen2022globally},{chen2022maneuvering}}
methodologies. Works in \cite{cao2019ratio},
\cite{kwon2022sign} provide more recent classifications on formation control methods as well as a comparative literature review on issues related to target formation’s constraints, required measurements, and shape convergence.

Within these categories, the position-based approach requires agents to possess a common understanding of a global coordinate system. Conversely, the displacement-based (often referred to as consensus-based) and bearing-based methods require that agents' local coordinate frames are perfectly aligned (have common orientation). 

Meanwhile, coordinate-free techniques, namely, distance- and angle-based methods in 
\cite{MEHDIFAR2020109086},
\cite{summers2009control}, \cite{jing2019angle},
\cite{chen2023angle}, and \cite{cao2019ratio} present a more attractive architecture for formation control due to their reduced implementation complexities and their ability to characterize the desired formation shape by a set of coordinate-free scalar variables, typically involving distances or angles. These scalar variables serve to define formation errors for the individual agents. Furthermore, agents must possess measurements of vectorized relative information of their neighboring agents (e.g., relative positions or bearings) in their local coordinate frames to constitute a control law. Thus, coordinate-free approaches facilitate the design and implementation of formation control laws within the confines of agents' local coordinate frames, obviating the necessity for global position measurements, such as those provided by GPS systems, or the presumption that agents' local coordinate frames are aligned. 
Another significant benefit of coordinate-free formation control strategies is the cost-effectiveness for agents. This is because they mandate simpler sensing and interaction mechanisms. While a majority of coordinate-free techniques depend on relative position measurements for all agents, only a few existing methods use only bearing or vision-based measurements  \cite{chen2023angle}, \cite{cao2019ratio}, \cite{mehdifar20222}, \cite{chen2020angle}. Vision based measurements are more practical since bearing data is captured straightforwardly using onboard cameras, making it more advantageous in real-world settings \cite{tron2016distributed}.

It is worth noting that the majority of existing research on coordinate-free formation control operates under the assumption of bidirectional sensings. This line of research often relies on various graph rigidity concepts, including distance, angle, ratio-of-distances, sign, and weak rigidity notions 
 \cite{{jing2019angle}, {chen2023angle}, {cao2019ratio}, {kwon2020generalized}, {kwon2022sign}}.
However, from a practical standpoint, it is more realistic to consider directed sensing among agents due to inherent sensing limitations or issues introduced by measurement mismatches in undirected formation control\cite{mou2015undirected}. 
In this regard, the concept of persistent graphs emerged as the directed counterpart of distance rigidity \cite{yu2007three}. Earlier control designs for achieving persistent formations can be found in \cite{summers2009control}. 
Regrettably, most coordinate-free formation control methods offer guarantees of local (non global) convergence to the desired shape. This means that even if agents eventually meet the desired formation constraints, they might not converge to the desired shape due to issues like reflection, flip, and flex ambiguities as highlighted in 
\cite{MEHDIFAR2020109086,chen2020angle,cao2019ratio,jing2019angle, kwon2020generalized,kwon2022sign}.

To deal with the issue of ambiguous shapes, several recent studies have proposed 2D and 3D distance-based formations by incorporating additional formation constraints to agents, such as signed areas/angles
and volumes, to allow defining the target shape uniquely \cite{kwon2022sign},
\cite{sugie2020global,liu2019directed,cao2019almost,anderson2017formation,babazadeh2022directed}.
However, these approaches inadvertently introduce undesired equilibria at particular agent positions through the control design procedure, which significantly complicates the controllers' gain adjustment for ensuring global shape convergence \cite{sugie2020global}, \cite{liu2019directed}, \cite{cao2019almost}.
Recent studies have introduced an alternative approach to guarantee global convergence of coordinate-free formations.
For instance, \cite{liu2020distance}, \cite{liu2021orthogonal}, utilize formation error variables along orthogonal directions to characterize directed 2D and 3D distance-based formations with guaranteed (almost) global convergence to the desired shape, respectively. Furthermore, \cite{chen2022globally} has provided global stabilization for the angle-rigid formations. However, agents must use relative position information. Moreover, recently \cite{mehdifar20222} proposed a 2D directed formation control approach using orthogonal bipolar coordinate variables to achieve almost global convergence to the desired shape. The main advantages of \cite{mehdifar20222} w.r.t. \cite{liu2020distance} and \cite{chen2022globally} are in employing bearing measurements instead of relative position measurements for follower agents and also providing extra degrees of freedom for adjusting scale and orientation of the formation. Nevertheless, the results in \cite{chen2022globally} and\cite{mehdifar20222} are only limited to 2D formations.

In this work, inspired by \cite{mehdifar20222}, we present a 3D directed coordinate-free formation control method using bispherical coordinates. Our method ensures (almost) global convergence to the desired shape while using only vision-based sensing for follower agents.
We utilize triangulated acyclic minimally persistent graphs to model the inter-agent sensing topology, which gives rise to a distance-rigid directed leader-follower structure with a minimal number of edges.
Moreover, our novel approach utilizes local bispherical coordinates to characterize formation errors, relying on onboard vision sensors for bearing and distance ratio measurements. This technique circumvents the need for relative position measurements, which are often challenging to obtain in environments like deep space.
Leveraging the fact that each follower’s formation errors can be reduced (independently) by moving along the three orthogonal directions of its associated bishperical coordinate basis, we design decentralized controllers achieving (almost) global asymptotic convergence to the target shape. To the best of the authors' knowledge, this work is the first to employ bispherical coordinates for coordinate-free 3D formation control for achieving (almost) global shape convergence without encountering undesired equilibria.
%%%%%%%%%%%%%%%%%%%%%%%%%%%%%%%%%%%%%%%%%%%%%%%%%%%%%%%%%%%%%%%%%%%%%%%%%%%%%%%%
\section{Problem Formulation} 
Consider a multiagent system composed of $n$ mobile agents in 3D, governed by the following dynamics:
\begin{equation}
\dot{p}_i=u_i,\quad i = 1,\dots,n,
 \label{dynamics}
\end{equation}
where $p_i \in \mathbb{R}^3$ and $u_i \in \mathbb{R}^3$ are the position and the velocity-level control input of agent $i$ expressed with respect to a global coordinate frame, respectively. 
Let us model the sensing topology among agents as a directed graph $\mathcal{G}=(\mathcal{V}, \mathcal{E})$, where $\mathcal{V}=\{1,2, \ldots, n\}$ is the set of vertices representing the agents and $\mathcal{E}=\{(j, i) \mid j, i \in \mathcal{V}, j \neq i\}$, where if $(j, i) \in \mathcal{E} \Rightarrow(i, j) \notin \mathcal{E}$, is the set of directed edges modeling the directed sensing among the agents. To be more precise, $(j, i) \in \mathcal{E}$ denotes an edge that starts from vertex $j$ (source) and sinks at vertex $i$, and its direction is indicated by $j \rightarrow i$. For $(j, i)$, we say $i$ is the neighbor of $j$.  The relative position vector ${p}_{ij}$, and the \textit{relative bearing} vector $\hat{v}_{ij} \in \mathbb{R}^3$ corresponding to the directed edge $(j, i)$ are defined as:
\begin{equation}
p_{j i}:=p_i-p_j, \quad \hat{v}_{ij}\coloneqq \frac{p_i-p_j}{\lVert{p_i-p_j}\rVert},\quad (j, i) \in \mathcal{E}.
\label{v_hat}
\end{equation}
Particularly, in this article, the physical interpretation of the directed edge $(j, i) \in \mathcal{E}$ is that only agent $j$ can measure the relative bearing of agent $i$ with respect to itself, i.e., $\hat{v}_{j i}$, and not vice versa. As will be highlighted later, we will also assume that only agent 2 is capable of measuring its relative position w.r.t. its neighbor, which is agent 1. 
Moreover, we assume that the graph $\mathcal{G}$ is triangulated and imposes a hierarchical structure, where agent 1 is the \textit{leader}, agent 2 is the \textit{first follower} with agent 1 acting as its only neighbor, agent 3 is the \textit{second follower} following agents 1 and 2, and agents $i \geq 4$ are \textit{ ordinary followers} with each one having exactly three neighbors to follow with smaller indices. Hence, we impose the following assumption for constructing $\mathcal{G}$.
\begin{assumption} \label{assum:G} The directed sensing graph $\mathcal{G}$ is constructed such that:
\begin{enumerate} 
\item[(i)]  $\operatorname{out}(i)=i-1, \forall i \leq 3$, and $\operatorname{out}(i)=3, \forall i \geq 4$;
\item[(ii)] if there is an edge between agents $i$ and $j$, where $i<j$, the direction must be $j \rightarrow i$;
\item[(iii)] if $(k, i),(k, j) \in \mathcal{E}$ then $(j, i) \in \mathcal{E}$.
\end{enumerate}
\end{assumption}
Here, out$(i)$ denotes the out-degree of vertex $i$ that is the number of edges in $\mathcal{E}$ whose source is vertex $i$ and whose sinks are in $\mathcal{V} \backslash\{i\}$.
It is important to note that (i) and (ii) of Assumption 
\ref{assum:G} impose $\mathcal{G}$ to be acyclic minimally persistent with edge set cardinality $|\mathcal{E}|=3 n-6$ \cite{anderson2008rigid}.
Moreover, Assumption 
\ref{assum:G}.(iii) ensures that $\mathcal{G}$ is triangulated and composed of acyclic-directed tetrahedrons. Fig. \ref{fig:framework-general} shows an example of $\mathcal{G}$ constructed under Assumption 
\ref{assum:G}.
\begin{remark}
It is known that graphs satisfying Assumption \ref{assum:G} can be systematically constructed using 3D Henneberg type I insertion \cite{babazadeh2022directed}, \cite{liu2021orthogonal}. Such graphs belong to a class of 
acyclic minimally persistent graphs,
which are the directed counterpart of undirected distance rigid graphs \cite{anderson2008rigid}, \cite{summers2009control}, \cite{yu2007three}. Minimally persistent graphs (persistent graphs with a minimum number of edges) are favorable in practice since they require a minimum number of relations (sensing) among agents.
\end{remark}

Associated with each ordinary follower $l\geq4$ in $\mathcal{G}$ with three neighbors $i$, $j$, and $k$, we can define a signed volume of the tetrahedron formed by vertices $i<j<$ $k<l$  as follows\cite{liu2021orthogonal}:
\vspace{-1mm}
% \cite{mallison1935use}:
\begin{equation}
\begin{aligned}
V_{ijkl} & =\frac{1}{6} \operatorname{det}\left[\begin{array}{cccc}
1 & 1 & 1 & 1 \\
p_i & p_j & p_k & p_l
\end{array}\right] \\
& =-\frac{1}{6}\left(p_i-p_l\right)^{\top}\left[\left(p_j-p_l\right) \times\left(p_k-p_l\right)\right].
\label{volume}
\end{aligned}
\end{equation}
The sign of $V_{ijkl}$ is interpreted as follows: If an observer positioned at vertex $l$ observes the sequence of vertices $i, j$, and $k$ in a counterclockwise orientation with respect to the plane containing $i, j$, and $k$, denoted as $\parallelogram _{i j k}$, the sign of $V_{ijkl}$ is positive. Conversely, a clockwise observation yields a negative sign for $V_{ijkl}$. Note that this volume metric becomes zero if any triad of vertices (i.e. $i,j,k$) becomes collinear or if all four vertices lie on the same plane. 
We define the stacked signed volumes corresponding to all tetrahedral sub-graphs of $\mathcal{G}$ by the mapping $\mathbf{V}$ : $\mathbb{R}^{3 n} \rightarrow \mathbb{R}^{n-3}$:
\begin{equation}
 \begin{aligned}
& \mathbf{V}(p)=\left[\ldots, \frac{1}{6} \operatorname{det}\left[\begin{array}{cccc}
1 & 1 & 1 & 1 \\
p_i & p_j & p_k & p_l
\end{array}\right], \ldots\right]^\top, \\
& \forall(l, i),(l, j),(l, k) \in \mathcal{E}-\{(2,1),(3,1),(3,2)\},
\end{aligned}   
\end{equation}
where $p=\left[p_1, \ldots, p_n\right]$ and its $m$-th ($m\leq n-3$) component is the signed volume of the $m$-th tetrahedron constructed with vertices $i<j<$ $k<l$. For example, the stacked signed volumes of the graph in Fig.\ref{fig:framework-general} has two elements, where $V_{1234}>0$  (shown in blue) and $V_{1345}<0$ (shown in red).
\begin{figure}[!tbp]
	\centering
%	\flushleft
	\begin{subfigure}[t]{0.35\linewidth}
		\centering
		\includegraphics[width=\linewidth]{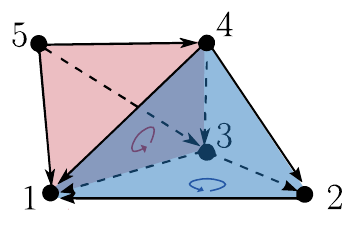}
		\caption{}
		\label{fig:framework-general}
	\end{subfigure}%
	~
	\begin{subfigure}[t]{0.3\linewidth}
		\centering
		\includegraphics[width=\linewidth]{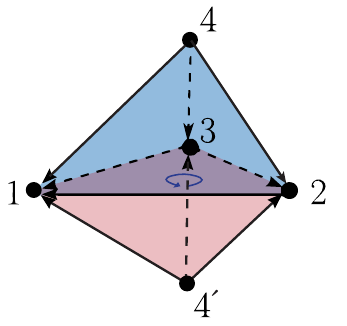}
		\caption{}
		\label{fig:reflected}
	\end{subfigure}%
 ~
 \begin{subfigure}[t]{0.3\linewidth}
		\centering
		\includegraphics[width=\linewidth]{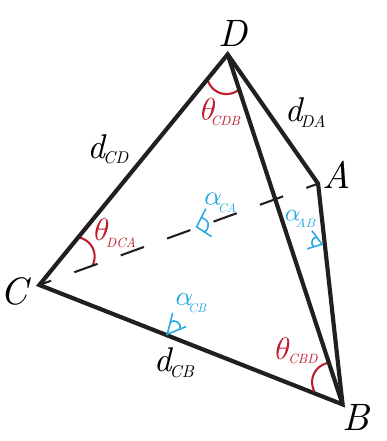}
		\caption{}
		\label{fig:tetrahedron}
	\end{subfigure}
 \caption{(a) A graph constructed under Assumption 1 with 5 agents, which consists of two tetrahedral subgraphs. The blue tetrahedron has a positive volume $(V_{1234}>0)$ and the red one has a negative volume $(V_{1235}<0)$. (b) Position of vertex 4 makes a tetrahedron with positive volume $V_{1234}>0$ (blue) while its reflected position 4' leads to the same volume ($|V_{1234}|=|V_{1234'}|$) with a negative sign $V_{1234'}<0$ (red). (c) Tetrahedron $ABCD$ and three of its edge lengths, face angles, and dihedral angles.\vspace{-0.4cm}}
	\label{fig:frameworks}
\end{figure}
 It is known that based on the directed sensing graph $\mathcal{G}$ under Assumption 
\ref{assum:G}, we can \textit{uniquely} define a desired formation characterized by the followings\cite{liu2021orthogonal}:
\begin{enumerate}
    \item A set of $3 n-6$ \textit{desired distances} $d_{j i}^*$, appointed to the directed edges $(j, i) \in \mathcal{E}$.
    \item A set of $n-3$ \textit{desired signed volumes} $V_{ijkl}^* \in \mathbf{V}(p^*)$
\end{enumerate}
    % * \in \mathbf{V}(p)^*$
Given a desired formation characterized by a graph $\mathcal{G}$ (under Assumption 
\ref{assum:G}) and the corresponding sets of desired distances and signed volumes, the objective is to design a decentralized controller for \eqref{dynamics} such that:
\begin{subequations}
\begin{align}
\lVert{p_i(t)-p_j(t)}\rVert & \rightarrow d_{j i}^* \;\;\,\quad \text { as } \quad t \rightarrow \infty, \\
V_{ijkl}(t) & \rightarrow V_{ijkl}^* \quad  \text { as } \quad t \rightarrow \infty,
\end{align}
\label{objective}%
\end{subequations}
for all $(j, i) \in \mathcal{E}$ and all $V_{ijkl} \in \mathbf{V}$, respectively, while avoiding zero distance among all neighboring agents (i.e., $\left\|p_{j i}\right\| \neq 0, \forall(j, i) \in \mathcal{E}, \forall t \geq 0$).

Note that characterization of the desired formation just by a set of desired distances is not unique and suffers from local shape convergence and reflection issues due to the existence of undesired shapes (namely reflection, flip, and flex ambiguities) (see \cite{anderson2008rigid}, \cite{liu2019directed}, \cite{anderson2017formation} for examples). To tackle these issues, extra types of formation parameters (e.g., signed volume, signed area, edge-angle) have been recently employed along with the distances to characterize the desired formation uniquely, which is necessary for having global shape convergence \cite{sugie2020global}, \cite{cao2019almost}, \cite{anderson2017formation}, 
\cite{liu2021orthogonal}.  
Particularly, the notion of strong congruency was introduced in \cite{liu2020distance} and \cite{liu2021orthogonal} for distinguishing the desired shape of a formation 
from its reflected version. 
It is known that satisfaction of \eqref{objective} is equivalent to strong congruency \cite{liu2021orthogonal} between the actual formation of the agents and the desired formation (see \cite{liu2021orthogonal}, Lemma 1]). This means that if \eqref{objective} gets satisfied, the agents can achieve the desired formation only up to rotations and translations \cite{liu2021orthogonal}.
As an example, consider distances $\left\|p_{41}\right\|,\left\|p_{42}\right\|,\left\|p_{43}\right\|$,  and the signed volume $V_{1234}$ in Fig. \ref{fig:reflected}, where $V_{1234}>0$. Assuming vertex $4'$ as the reflected version of vertex $4$ without altering its distances with respect to $1$, $2$ and $3$, then we get $V_{1234'} < 0$. In general, this property allows us to distinguish the position of agent $l$ from its reflection with respect to the plane passing through agents $i$, $j$, and $k$.
%%%%%%%%%%%%%%%%%%%%%%%%%%%%%%%%%%%%%%%%%%%%%%%%%%%%%%%%%%%%%%%%%%%%%%%%%%%%%%%%
\section{Bispherical Coordinate System}
Consider a tetrahedron defined by vertices $A,B,C,D$. It encompasses 4 \textit{faces} $F_{\mathcal{P}\mathcal{Q}\mathcal{R}}$ for each triad of $\{\mathcal{P},\mathcal{Q},\mathcal{R}\} \in \{A,B,C,D\}$,  6 \textit{edge lengths} $d_{\mathcal{P}\mathcal{Q}}$ for each unique pair of vertices $\{\mathcal{P},\mathcal{Q}\} \in \{A,B,C,D\}$, 12 \textit{face angles} $\theta_{\mathcal{P}\mathcal{Q}\mathcal{R}}$ defined by the angles between edges ${\mathcal{P}\mathcal{Q}}$ and ${\mathcal{Q}\mathcal{R}}$ for each set of vertices $\{\mathcal{P},\mathcal{Q},\mathcal{R}\} \in \{A,B,C,D\}$, and 6 \textit{dihedral angles} $\alpha_{\mathcal{P}\mathcal{Q}}$ representing the angles between the faces adjoining edge ${\mathcal{P}\mathcal{Q}}$ for each vertex pair $\{\mathcal{P},\mathcal{Q}\} \in \{A,B,C,D\}$.
Fig.\ref{fig:tetrahedron} illustrates three of the edge lengths, face angles, and dihedral angles on tetrahedron $ABCD$.
\subsection{Bispherical Coordinates}
Now, consider a tetrahedral subgraph of $\mathcal{G}$ where $i<j<k<l$. The first bispherical coordinate, denoted by $\xi_l \in [0,\pi]$, is equal to the edge angle $\theta_{ilj}$, which can be expressed as the angle between relative bearing vectors from agent $l$ to the neighboring agents $i$ and $j$ as follows:
\begin{flalign}
    \xi_l:=\theta_{ilj}&=\cos ^{-1}\left(\hat{v}_{li}^\top \hat{v}_{lj}\right),\quad l\in \mathcal{V} \backslash\{1,2,3\}.
 \label{xi_l}
\end{flalign}
Note that for agent 3 we define: $\xi_{3}:=\cos ^{-1}\left(\hat{v}_{31}^\top \hat{v}_{32}\right)$.
Fig. \ref{frame3d} illustrates the relative bearing vectors among the agents in a tetrahedral subgraph and the first bispherical coordinate $\xi_k$ and $\xi_l$ for agents $k$ and $l$, respectively.
Defining the \textit{ratios of distances} 
$r_{lij} :=\left\|p_{l i}\right\|/\left\|p_{l j}\right\|=d_{l i} /d_{l j}$ and $r_{lik} :=\left\|p_{l i}\right\|/\left\|p_{l k}\right\|=d_{l i}/d_{l k}$ 
for $l \in \mathcal{V} \backslash\{1,2,3\}$, one can define the second bispherical coordinate $\eta_l$ as:
\vspace{-1mm}
\begin{equation}
\eta_l:= \ln r_{lij} = \ln \frac{\left\|p_{l i}\right\|}{\left\|p_{l j}\right\|}, \quad l \in \mathcal{V} \backslash\{1,2,3\},
\label{eta}
\vspace{-1mm}
\end{equation}
where $\eta_l \in \mathbb{R}$. Note that, when agent $l$ approaches agent $i$ or agent $j$ (i.e., either $d_{li} \rightarrow 0$ or $d_{li} \rightarrow 0$ ), $\eta_l$ tends to $\pm \infty$. Note that, only one ratio of the distance is defined for agent 3, that is $\eta_3 = \ln(d_{31}/d_{32})$.
Finally, the third bispherical coordinate, denoted by $\varphi_l \in [0,2\pi) ,\; l \geq4$,
is the angle between half-planes $\parallelogram_{ijk}$ and $\parallelogram_{ijl}$ measured in the counterclockwise direction from the former to the latter (see Fig. \ref{frame2d}). 
In particular, one can obtain $\varphi_l$ as follows:
\begin{equation}
\resizebox{.88\linewidth}{!}{$
\varphi_l= \begin{cases}\alpha_{ij}& \text { if }\mathrm{sgn}\left(\hat{v}_{li}^\top\left(\hat{v}_{lj} \times \hat{v}_{lk}\right)\right)=\mathrm{sgn}\left(V_{ijkl}\right)>0 \\ 2 \pi-\alpha_{ij} & \text { if }\mathrm{sgn}\left(\hat{v}_{li}^\top\left(\hat{v}_{lj} \times \hat{v}_{lk}\right)\right)=\mathrm{sgn}\left(V_{ijkl}\right)<0\\ \pi & \text { if } V_{ijkl}=0 \text{ and } \left(\hat{v}_{j i} \times\hat{v}_{li}\right) ^\top\left(\hat{v}_{j i} \times \hat{v}_{ki}\right) <0\\ 0 &\text{ otherwise } \end{cases}
\label{phi_l}
$}\;,
\end{equation}
where $\alpha_{ij} \in (0,\pi)$ is the dihedral angle of the tetrahedron $ijkl$ on edge $(j,i)$. 
Note that the sign of $\hat{v}_{li} ^\top\left(\hat{v}_{lj} \times \hat{v}_{lk}\right)$ is the same as the sign of $V_{ijkl}$ since $\hat{v}_{li} ^\top\left(\hat{v}_{lj} \times \hat{v}_{lk}\right)$ consists of the normalized vectors used in \eqref{volume}.
The third and fourth cases in \eqref{phi_l}
correspond to when agent $l$ and all its neighbors are in the same plane meaning that $V_{ijkl}=0$. As $\alpha_{ij}$ is undefined in these cases, if the half-plane including agents $i$, $j$, and $k$ is the same half-plane including agents $i$, $j$, and $l$ then $\varphi_l =0$ . Otherwise, $\varphi_l =\pi$. Moreover, when 3 or 4 number of agents in the subgraph with agents $i$, $j$, $k$, and $l$ are collinear, $\varphi_l=0$. Finally, note that for agent 3 we always have $\varphi_3 = 0$ since it always lies on $\parallelogram_{123}$.
Also, note that one can find $\beta_{ij}:=\cos\alpha_{ij}$ by calculating the angle between the normal vectors of faces $F_{ijk}$ and $F_{ijl}$ as follows:
\begin{equation}
\begin{aligned}
        \beta_{i j}&=\frac{\left(\hat{v}_{i j} \times \hat{v}_{i k}\right)^\top\left(\hat{v}_{i j} \times \hat{v}_{i l}\right)}{\lVert\hat{v}_{i j} \times \hat{v}_{i k}\rVert\lVert\hat{v}_{i j} \times \hat{v}_{i l}\rVert}\\&=\frac{\cos \theta_{k i l}-\cos \theta_{j i k} \cos \theta_{j i l}}{\sin \theta_{j i k} \sin \theta_{j i l}}.
    % ,\quad 0<\alpha_{i j}<\pi.
    \end{aligned}
    \label{alpha}
\end{equation}
\subsection{Desired Formation Characterization}
Given a target formation,
% expressed by the graph $\mathcal{G}$ along with the desired signed volumes and distances
one can use the desired bispherical coordinates of agent $m\geq3$ with respect to its neighbors to uniquely characterize the desired formation shape using the following expressions:
\begin{flalign}  
&\eta_l^*=\ln \frac{d_{l i}^*}{d_{l j}^*},\nonumber\\&\xi_l^*=\theta_{ilj}^*=\cos^{-1}\left( \frac{{d_{li}^*}^2+{d_{l j}^*}^2-{d_{ji}^*}^2}{2d_{l i}^*d_{l j}^*}\right),\label{bispherical_star}\\
&\varphi_l^*= \begin{cases}\alpha_{ij}^* = \cos^{-1}\left(\frac{\cos \theta_{k i l}^*-\cos \theta_{j i k}^* \cos \theta_{j i l}^*}{\sin \theta_{j i k}^* \sin \theta_{j i l}^*}\right)& \text { if }V_{ijkl}^*>0 \\ 2 \pi-\alpha_{ij}^* & \text { if }V_{ijkl}^*<0 \end{cases}\nonumber
\end{flalign}
Note that, for agent 3 we have:
\begin{equation}
\eta_3^*:=\ln \frac{d_{31}^*}{d_{32}^*},\quad\xi_3^* := \cos^{-1}\left( \frac{{d_{31}^*}^2+{d_{32}^*}^2-{d_{21}^*}^2}{2d_{31}^*d_{32}^*}\right).    
\end{equation}
\begin{lemma}\label{lem:objective}
   Given a desired formation shape based on a specific directed sensing graph $\mathcal{G}=(\mathcal{V}, \mathcal{E})$ under Assumption \ref{assum:G}, satisfying
   \end{lemma}
\begin{subequations}
\begin{align} 
\left\|p_2(t)-p_1(t)\right\| \rightarrow d_{21}^*, & & \text { as } t \rightarrow & \infty \label{dist_goal}\\
\xi_m(t) \rightarrow \xi_m^*,& \quad  m \geq 3, & \text { as } t \rightarrow & \infty \label{xi_obj}\\
\eta_m(t) \rightarrow \eta_m^*,& \quad  m \geq 3, & \text { as } t \rightarrow &\infty \label{eta_obj}\\
\varphi_m(t) \rightarrow \varphi_m^*,& \quad  m \geq 4, & \text { as } t \rightarrow & \label{phi_obj}\infty 
\end{align}
\label{bispherical_objective}
\end{subequations}
is equivalent to the satisfaction of \eqref{objective}.

\begin{proof}
    % See Appendix \ref{appen:proof_congruent}.
    Consider a formation of agents with the sensing graph $\mathcal{G}$ constructed under Assumption \ref{assum:G}. Let $p$ be the stacked positions of agents and define the stacked vector of formation control variables as $\Gamma(p)=\left[\Gamma_2, \Gamma_3, \Gamma_4, \ldots, \Gamma_n\right]$, where $\Gamma_i$ is the formation control variables associated to agent $i$. More precisely we have $\Gamma_2  = d_{21}$, $\Gamma_3  =\left[\xi_3, \eta_3\right]$, and $\Gamma_m  =\left[\xi_m, \eta_m, \varphi_m\right],\; m \geq 4$. Assume that the desired formation is characterized by desired distances $d^*_{ji},\;(j,i) \in \mathcal{E}$ and desired signed volumes $V_{ijkl}^* \in \mathbf{V}(p^*)$, where $p^\ast$ denotes the stacked vector of the desired positions of agents. In this respect, one can define $\Gamma(p^*)=\left[\Gamma_{2}^*, \Gamma_{3}^*, \Gamma_{4}^*, \ldots, \Gamma_{n}^*\right]$, where $\Gamma_{2}^*  = d^*_{21}$, $\Gamma_{3}^*  =\left[\xi^*_3, \eta^*_3\right]$, $\Gamma_{m}^*  =\left[\xi^*_m, \eta^*_m, \varphi^*_m\right],\; m \geq 4$.
Now, we argue that the distances and signed volumes of the formation with graph $\mathcal{G}$ are equal to their counterpart in the desired formation if and only if $\Gamma(p)=\Gamma(p^*)$. 

\textbf{Only if: } Since all edges in the edge set are the same in the graph, it is obvious that $\Gamma_2 = \Gamma_{2}^*$. Due to the same reason, triangles on faces $F_{123}$ and $F_{1^*2^*3^*}$ are similar meaning that $\theta_{132} = \theta^*_{123}$ and consequently $\xi_3 = \xi_{3}^*$. Also, due to the definition of $\eta$ for agent $3$, $\eta_3 = \eta_{3}^*$. This means that $\Gamma_3 = \Gamma_{3}^*$. Now consider the tetrahedrons $1234$ and $1^*2^*3^*4^*$. Similarily, all 12 face angles in 4 faces of $1234$ are equal to their counterpart in $1^*2^*3^*4^*$. This means $\eta_4 = \eta_{4}^*$, $\xi_4 = \xi_{4}^*$, and $\alpha_{12} = \alpha_{12}^*$ according to \eqref{alpha}. One can show the following relation between the signed volume and the third bispherical coordinate holds for agent $l \geq 4$ \cite{liu2020orthogonal}:
\begin{equation}
    \sin \varphi_l = \frac{6V_{ijkl}d_{ji}}{\lVert p_{j i} \times p_{ki}\rVert\lVert p_{li} \times p_{lj}\rVert}.
    \label{volume_dihedral}
\end{equation}
 Now, due to the equality of face angles and edge angles and the signed volumes of the tetrahedrons, $\sin\varphi_4 = \sin \varphi_{4}^*$. This proves that $\varphi_4 = \varphi_{4}^*$, hence, $\Gamma_4 = \Gamma_{4}^*$. Following the same argument for agents $4<m\leq n$ will lead to $\Gamma(p)=\Gamma(p^*)$. 

\textbf{if: } For agent 2, it is obvious that $d_{21} = d_{21}^*$. Now, since $\xi_{3} = \xi_{3}^*$ one can show the following by the application of the cosines law:
\begin{equation}
    \frac{1+e^{2\eta_3}-\left(\frac{d_{21}}{d_{32}}\right)^2}{2 e^{\eta_3}}=\frac{1+e^{2\eta_{3}^*}-\left(\frac{d_{21}^*}{d_{32}^*}\right)^2}{2 e^{\eta_{3}^*}}.
    \label{eta_cosine}
\end{equation}
Now, knowing $d_{21} = d_{21}^*$ and $\eta_{3} = \eta_{3}^*$, one can show that $d_{32} = d_{32}^*$ and $d_{31} = d_{31}^*$. Similarly, one can prove $d_{41} = d_{41}^*$ and $d_{42} = d_{42}^*$. Thus, all face angles in $F_{123}$ and $F_{124}$ are equal to their counterpart in $1^*2^*3^*4^*$. Subsequently, one can show that $\theta_{134} = \theta_{134}^*$ using \eqref{alpha}. From \eqref{volume_dihedral}, since all lengths and face angles are equal to their counterparts we get $V_{1234} = V_{1234}^*$. Then, applying side-angle-side similarity theorem to faces $F_{134}$ and $F_{134}^*$, one can show $d_{43}=d_{43}^*$. Repeating this argument for all tetrahedral subgraphs will lead to the equality of all edge lengths in the edge sets as well as the equality of stacked signed volumes. 
\end{proof}

Applying the same reasoning as mentioned in the proof of Lemma \ref{lem:objective}, it can be demonstrated that if $\Gamma_{2} = \sigma \Gamma_2^*$, where $\sigma$ represents the scale factor and $\Gamma_{i} = \Gamma^*_{i}, \; 3\leq i \leq n$, all distances and signed volumes within the formation will increase or decrease by a factor of $\sigma$ in comparison to their counterparts in the desired formation. This implies that the entire formation can be resized by adjusting the distance between agent 2 (the first follower) and agent 1 (the leader).

Recall that based on \eqref{objective}, only the first follower (agent 2) and second follower (agent 3) have to keep specific distances with respect to their neighbors, while all other agents (ordinary followers) must preserve a designated signed volume and three exact distances relative to their neighboring agents. Thus, to achieve the target formation through satisfying \eqref{objective}, each ordinary follower ($m \geq 4$) is required to control four variables: three distances and a signed volume. Nevertheless, introducing an additional shape constraint (i.e. signed volume) for the ordinary followers may cause new undesirable equilibria due to the interaction of distance and signed-volume constraints at particular agent locations (refer to \cite{sugie2020global}, \cite{liu2019directed},
\cite{cao2019almost}, \cite{anderson2017formation} for an in-depth discussion and examples in 2D formations). It is crucial to note that these four variables do not always form an orthogonal space, in which each one can be controlled independently via moving along orthogonal directions. Lemma \ref{lem:objective} circumvents this by requiring ordinary followers to control merely three orthogonal (independent) formation variables \eqref{eta_obj}, \eqref{xi_obj}, and \eqref{phi_obj}. We will exploit this property to design decentralized formation controllers for the follower agents, as detailed in Section \ref{controller-sec}, thereby allowing (almost) global convergence to the desired shape. Moreover, as mentioned earlier, by altering the distance between agents 2 and  1, $\left\|p_2(t)-p_1(t)\right\|$, the formation's scale at steady-state can be controlled. Here, scaling refers to maintaining all angles while adjusting all edge lengths proportionally. Thus, by dynamically setting a target distance $d_{21}^*(t)$ relative to the leader, the first follower can modulate the formation's scale, which is vital in real-world formation control scenarios like navigating through tight spaces or avoiding obstacles. 
\subsection{Bispherical Coordinates Basis Vectors}
In the previous subsections we showed that the desired positions of agents $l\geq4$ w.r.t. their neighbors can be uniquely characterized by  bispherical coordinates. Here, we derive the bispherical coordinate basis associated with each follower agent $l\geq4$.

Note that, in each tetrahedral subgraph of $\mathcal{G}$, where $i<j<k<l, \;l\geq4$, one can define a virtual local Cartesian coordinate frame for agent $l$ 
 denoted by $\left\{C_l\right\}$, with its origin located in the middle of the $i-j$ line (see Fig. \ref{frame3d}). 
The basis of $\{C_l\}$ can be written in terms of the relative bearing vectors (expressed in a global coordinate frame) associated with agent $l$ as follows:
 \begin{equation}
    \hat{X}_l=-\hat{v}_{ji},\quad \hat{Z}_{l}=\frac{\hat{v}_{ji}\times\hat{v}_{ki}}{\lVert\hat{v}_{ji}\times\hat{v}_{ki}\rVert},\quad \hat{Y}_l=\hat{Z}_{l}\times\hat{X}_l.
    % ,\quad l\geq4
    \label{cartesian}
\end{equation}
The bispherical coordinates are related to the $\left\{C_l\right\}$ frame with the following (almost) one-to-one (except at the foci of the bispherical coordinates, $i$ and $j$) transformation \cite{ARFKEN_1970}:
\begin{equation} 
\begin{gathered}
    x_l^{\left[C_l\right]} =\frac{a_l \sinh \eta_l}{\cosh \eta_l-\cos \xi_l},\quad
y_l^{\left[C_l\right]} =\frac{a_l \sin \xi_l\cos \varphi_l}{\cosh \eta_l-\cos \xi_l},\\
z_l^{\left[C_l\right]} =\frac{a \sin \xi_l \sin \varphi_l}{\cosh \eta_l-\cos \xi_l},
% \end{gathered}
\label{bi2cart}
\end{gathered}
\end{equation}
where $p_l^{[C_l]}=[x_l^{[C_l]}, y_l^{[C_l]},z_l^{[C_l]}]^\top \in \mathbb{R}^3$
is the position of vertex $l$ with respect to frame $\left\{C_l\right\}$ and $a_l:=0.5\left\|p_{j i}\right\|>0$.
The bispherical coordinate system $\left( \xi_l, \eta_l, \varphi_l\right)$ is indeed a 3D orthogonal curvilinear coordinate system \cite{ARFKEN_1970},\cite{Moon_Spencer_1988} (similar to the spherical coordinate system), thus, a local orthogonal basis can be defined at each point in the 3D plane of $\left\{C_l\right\}$ showing the directions of increase for $\xi_l$, $\eta_l$, and $\varphi_l$. Figs. \ref{frame} and \ref{isoquant} altogether show orthogonal bispherical coordinates basis  $\hat{\xi}_l$, $\hat{\eta}_l$, and $\hat{\varphi}_l \in \mathbb{R}^3$ associated with $\left\{C_l\right\}$ at some arbitrary points of interest.
\begin{remark}
  The relations in \eqref{cartesian} and \eqref{bi2cart} are still valid for agent 3. Indeed since agent 3 is always on $\parallelogram_{123}$ plane, we have $z_3^{[C_l]}=0$ and $\varphi_3=0$. Moreover, when agents $i$, $j$, and $k$ are collinear or collocated, the basis in \ref{cartesian} are not well-defined. To tackle this issue, we can use Algorithm 1 of \cite{liu2021orthogonal}, which guarantees that $\hat{v}_{ji}\times\hat{v}_{ki}$ is well defined. This means that the output vector of \cite[Algorithm 1]{liu2021orthogonal} will be an arbitrarily selected vector perpendicular to $p_{ji}$ and $p_{ki}$ when agents are collinear or collocated. Also, it is worth mentioning that the unit orthogonal basis of $\{C_l\}$ are the normalized version of the orthogonal basis defined in \cite{liu2021orthogonal}.  
\end{remark}
\begin{figure}[!tbp]
	\centering
	\begin{subfigure}[t]{.8\linewidth}
		\centering
		\includegraphics[width=\linewidth]{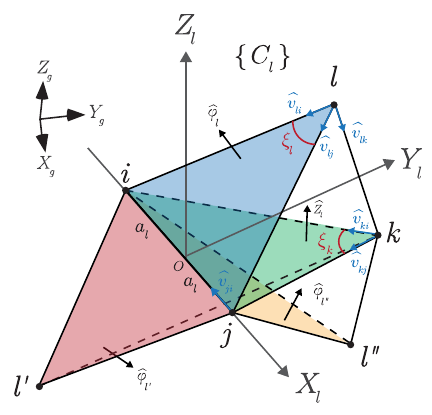}
		\caption{}
		\label{frame3d}
	\end{subfigure}\\ % Use "\\" to stack the figures vertically

	\begin{subfigure}[t]{.8\linewidth} % Adjust the width to match or to fit your layout needs
		\centering
		\includegraphics[width=\linewidth]{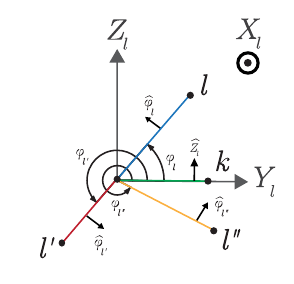}
		\caption{}
		\label{frame2d}
	\end{subfigure}
  	\caption{(a) Showing direction of $\hat{\varphi}$ for three cases of agent $l$ positions (indicated by $l$, $l'$ and $l''$).
  	Note that agents $i$, $j$, $l$, and $l'$ are on the same plane but different half-planes leading to opposite $\hat{\varphi}$ directions. (b) 2D view of $\{C_l \}$ showing $\hat{\varphi}$ and $\varphi$
  	for some different positions of agent $l$. (indicated by $l$, $l'$ and $l''$)} 
          \label{frame}
\end{figure}
  \begin{figure}[bpt!]
         \centering
         \includegraphics[width=.8\linewidth]{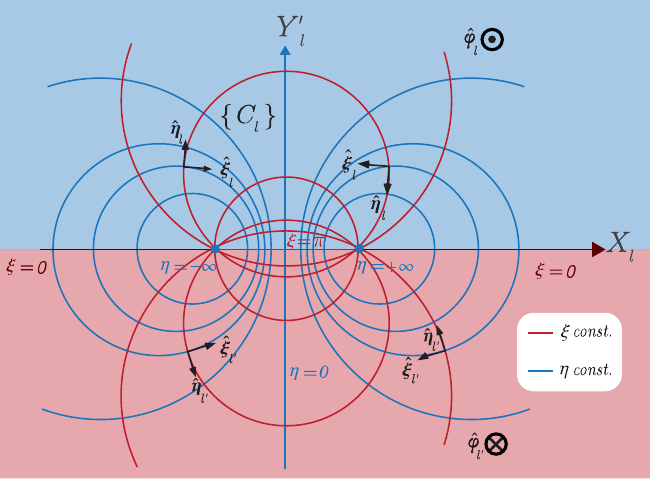}
         \caption{2D view of $i-j-l$ plane in Fig. \ref{frame3d}. In the entire upper half plane (blue) $\hat{\varphi}$ is pointing outwards, while it points inwards in the lower half plane (red). $Y^\prime_l$ is rotated version of $Y_l$ around $X_l$ lying down in $i-j-l$ plane. This figure depicts the projection of the isoquant surfaces of $\xi$, $\eta$ on $i-j-l$ plane and their corresponding orthogonal basis vectors $\hat{\xi}$ and $\hat{\eta}$ for two different positions of $l$ and $l'$. Note that $\hat{\xi}$ and $\hat{\eta}$ always lie on the $i-j-l$ plane, which are perpendicular to $\hat{\varphi}$.}
         \label{isoquant}
         \vspace{-0.5cm}
 \end{figure}
% \noindent\textit{Remark 4:}
\begin{remark}
Agent $l$ can calculate the cartesian basis vectors in Eq.\eqref{cartesian} via the relative bearing vectors of its neighbors. More precisely, by direct measurements of $\hat{v}_{li}$, $\hat{v}_{lj}$, $\hat{v}_{lk}$ and the ratio of the distances $r_{l i j}$, $r_{ljk}$, agent $l$ can obtain $\hat{v}_{ji}$ and $\hat{v}_{ki}$ without any communication with its neighbors using the following relations:
\begin{equation}
        \hat{v}_{k i}=\frac{r_{l i k} \hat{v}_{li}-\hat{v}_{lk}}{\lVert r_{l i k} \hat{v}_{li}-\hat{v}_{lk}\rVert},\quad\hat{v}_{j i}=\frac{r_{l i j} \hat{v}_{li}-\hat{v}_{lj}}{\lVert r_{l i j} \hat{v}_{li}-\hat{v}_{lj}\rVert}.
\end{equation}
\end{remark}
\vspace{3mm}
\begin{remark}
It is known that typical onboard vision-based sensors like monocular cameras provide projective measurements lacking distance data, yielding only bearing information which facilitates angle determination in \eqref{xi_l} and \eqref{phi_l}\cite{tron2016distributed}. Furthermore, as outlined in \cite{cao2019ratio}, Sec. II.D], the ratio of distance measurements can be deduced from a single camera image by comparing projections of two identical, differently sized objects or markers. Hence, in this work,  equipping followers with low-cost vision sensors suffices for necessary data acquisition, in contrast to many similar studies in coordinate-free formation control assuming relative position measurements for all agents\cite{chen2022globally,chen2022maneuvering,sugie2020global,liu2019directed,liu2020distance,liu2021orthogonal,cao2019almost}.
\end{remark}
\begin{lemma}
  For a given tetrahedral directed sub-graph as in Fig. \ref{frame3d}, the bispherical coordinates basis $\hat{\xi}_l, \hat{\eta}_l, \hat{\varphi}_l,\; l\geq 4$ (see Figs. \ref{isoquant} and \ref{frame}) associated with the virtual local Cartesian frame $\left\{C_l\right\},\; l\geq 4$ in Fig. \ref{frame} can be expressed with respect to $\{C_l\}$ as follows:
  % \vspace{-.5mm}
\begin{subequations} \label{basisvectors}
\begin{align}
\hat{\xi}_l&=f_1 \hat{X}_l+f_2 f_3 \hat{Y}_l+f_2 f_4 \hat{Z}_l, \label{xi_hat}\\
\hat{\eta_l}&=-f_2 \hat{X}_l+f_1 f_3 \hat{Y}_l+f_1 f_4 \hat{Z}_l, \label{eta_hat}\\
\hat{\varphi}_l&=-f_4 \hat{Y}_l+f_3 \hat{Z}_l,\label{phi_hat}
\end{align}
  % \vspace{-.5mm}
\label{hat_eq}%
\end{subequations}
where $\hat{X}_{l}$, $\hat{Y}_l$, and $\hat{Z}_l$ are given by \eqref{cartesian}, which are the Cartesian basis vectors assigned to agent $l$ in $\{C_l\}$ and
\begin{equation}
%    \resizebox{.75\linewidth}{!}
% {
\begin{gathered}
   f_1(\xi_l,\eta_l,\varphi_l)=\frac{-\sinh \eta_l \sin \xi_l}{\cosh \eta_l-\cos \xi_l},\; f_3(\varphi_l)=\cos\varphi_l,
\label{hat_1}% 
\end{gathered}
\end{equation}
\begin{equation}
%    \resizebox{.75\linewidth}{!}
% {
\begin{gathered}
    f_2(\xi_l,\eta_l,\varphi_l)=\frac{\cosh \eta_l \cos \xi_l-1}{\cosh \eta_l-\cos \xi_l},\; f_4(\varphi_l)= \sin\varphi_l.
    \label{hat_2}%
\end{gathered}
\end{equation}
\end{lemma}
\vspace{2mm}
\begin{proof}
    The proof is straightforward and can be followed similarly to \cite{ARFKEN_1970}.
\end{proof}
 Note that \eqref{hat_1} and \eqref{hat_2} are also valid for agent 3 since $\varphi_3 =0$, and $\hat{\varphi}_3 =\hat{Z}_3$.
\section{Proposed Controller}
\label{controller-sec}
\subsection{Formation Errors}
To quantify the control objectives we define 4 types of error variables. First, \textit{squared distance error} between agents 2 and 1 is defined as:
\begin{equation}
e_d=\left\|p_{21}\right\|^2-{d_{21}^{*2}}.
\label{error_distance}
\end{equation}
Second, the \textit{edge-angle} and \textit{logarithmic ratio of the distances errors} are defined as:
\begin{equation}
e_{\xi_m}=\xi_{m}-\xi_{m}^*, \quad e_{\eta_m}=\eta_m-\eta_m^*, \quad  m\geq3,
\label{error_xi}
\end{equation}
where 
$\xi_{m}^*$, and $\eta_m^*$ are given in 
\eqref{bispherical_star}.
Finally, the \textit{dihedral angle error} is defined as:
\begin{equation}
e_{\varphi_m}=\varphi_{m}-\varphi_{m}^*, \quad m\geq4,
\label{error_phi}
\end{equation}
where 
$\varphi_{l}^*$ is 
given in 
\eqref{bispherical_star}. 
Note that \eqref{error_xi} and \eqref{error_phi} are independent (orthogonal) error variables defined only for the followers. More precisely, by moving along each bispherical coordinate's basis, $\hat{\xi}_m$, $\hat{\eta}_m$, and $\hat{\varphi}_m$, each follower can reduce \eqref{error_xi} and \eqref{error_phi} respectively, without affecting the other error variable. Finally, due to the above discussion and Lemma \ref{lem:objective}, by adopting the bispherical coordinates approach, the control objective of \eqref{objective} is met by zero stabilization of the error signals defined in \eqref{error_distance}, \eqref{error_xi}, and \eqref{error_phi}.
\subsection{Proposed Control Law and Stability Analysis}
Notice that in the proposed formation control setup, the leader (agent 1) does not participate in forming the desired shape, thus its behavior is independent of the other agents. In this respect, the leader's control law is ignored without loss of generality.
We propose the following formation control laws:
\begin{subequations} \label{controllaws}
\begin{align}
u_1 & =0, \label{control_leader}\\
u_2 & =\kappa_2 e_d p_{21} = \kappa_2 e_d\left\|p_{21}\right\|\hat{v}_{21}, \label{control_second}\\
u_3 & =-\kappa_3 e_{\xi_3} \hat{\xi}_3-\lambda_3 e_{\eta_3}\hat{\eta}_3, \label{control_third}\\u_l & =-\kappa_l e_{\xi_l} \hat{\xi}_l-\lambda_l e_{\eta_l}\hat{\eta}_l - \gamma_l e_{\varphi_l}\hat{\varphi_l}, 
\quad l\geq4,
\label{control_follower}
\end{align}
\label{controller}%
\end{subequations}
where $\kappa_2$, $\kappa_3$, $\lambda_3$, and $\kappa_l,\lambda_l,\gamma_l,\; l\geq4$, are positive control gains and $\hat{\xi}_l$, $\hat{\eta}_l$, and $\hat{\varphi}_l$ are the bispherical coordinates basis associated with ordinary followers given in \eqref{basisvectors}. Note that the proposed control laws in \eqref{controllaws} are decentralized, since each agent only uses the measured information w.r.t. its neighbors. 
\begin{remark}
   Although the proposed control laws \eqref{controller} are given with respect to a virtual coordinate frame $\{C_l\}, \;l\geq3$ (only for the sake of analysis), we emphasize that the proposed formation controller can be implemented in any arbitrarily oriented local coordinate frame (i.e., in a coordinate-free fashion).
   Since all $e_h,\; h \in\left\{d, \eta_l, \xi_l, \varphi_l\right\}$ are all scalar variables (e.g. distance, angles) these errors are invariant to coordinate transformations. Now let the superscript $[l], l \geq 3$, indicate a quantity expressed in the local coordinate frame of the $l$-th agent. Furthermore, suppose that $\mathcal{R}_l \in \mathrm{SO}(3)$ is the transformation (rotation) matrix from the $l$-th local frame to the virtual frame. Notice that, we have $u_l=\mathcal{R}_l u_l^{[l]}, p_{l i}=\mathcal{R}_l p_{l i}^{[l]}=\mathcal{R}_l(p_i^{[l]}-p_l^{[l]})$, and consequently from \eqref{v_hat} we get $\hat{v}_{l i}=\mathcal{R}_l \hat{v}_{l i}^{[l]}, i<l \in \mathbb{N}$. Considering the definition of the virtual local Cartesian coordinate basis and the bispherical in \eqref{cartesian} and \eqref{hat_eq} used in \eqref{control_follower}, one can show that $u_l^{[l]}=\mathcal{R}_l^{-1} u_l  =-\kappa_l e_{\eta_l}\left(e_{\eta_l}\right) \hat{\eta}_l^{[l]}-\lambda_l e_{\xi_l}\left(e_{\xi_l}\right) \hat{\xi}_l^{[l]}-\gamma_l{\varphi_l}\left(e_{\varphi_l}\right) \hat{\varphi}_l^{[l]}$. Similarly, this claim can be also verified for agents 2 and 3.
\end{remark}

 The following theorem summarizes our main results. 
 \begin{theorem}\label{theorem:main}
    Consider a group of $n$ agents with dynamics \eqref{dynamics} in a 3D space. Let the desired formation be defined by a directed graph $\mathcal{G} = (\mathcal{V},\mathcal{E})$ under Assumption \ref{assum:G} along with the sets of desired formation parameters $d_{21}^*$, $\xi_{m}^*,\; \eta_{m}^*, \; m \geq 3$, and $\varphi_{m}^*,\; m \geq 4$. The decentralized control protocol \eqref{controller} ensures $\left\|p_{j i}\right\| \neq 0, \forall(j, i) \in \mathcal{E}, \forall t \geq 0$, and renders the formation errors in \eqref{error_xi} and \eqref{error_phi} almost globally asymptotically stable, which gurantees the satisfaction of the desired objectives in \eqref{bispherical_objective}.
\end{theorem}
\begin{proof}
The formation error dynamics can be regarded as a cascade connected system and the proof follows the stability analysis of cascaded systems. By proving almost Global Asymptotic Stability (a-GAS) of every outer subsystem, we establish the cascaded formation error system is a-GAS, leading to the satisfaction of \eqref{bispherical_objective}. The formation errors' dynamics corresponding to each agent are derived in Appendix \ref{error_dynamics}.

First, consider the two-agent subsystem including agents 1 and 2. Using the Lyapunov function candidate $W_2 \coloneqq e_d^2/2$, where $e_d$ is defined in \eqref{error_distance}, one can show that this subsystem under control inputs \eqref{control_leader}
and \eqref{control_second} is a-GAS if $p_{21}(0) \neq 0$, that is $e_d = -{d_{21}^*}^2$ (see \cite{liu2020distance} for a detailed analysis). Now if agent 3 is added to this subsystem, we get the following cascade interconnection:
\begin{equation}
\begin{aligned}
& \dot{e}_3=f_3\left({e}_3, e_2\right), \\
& \dot{e}_2=g_2\left(e_2\right),
\end{aligned}
\label{cascade-three}
\end{equation}
where $e_3 \coloneqq [e_{\xi_3},e_{\eta_3}]$ is the vector of stacked formation errors for agent 3 and $e_2 \coloneqq e_d$. Note that the exact expressions of $f_3$ and $g_2$ can be found by substituting control laws \eqref{control_second} and \eqref{control_third} into \eqref{distance_dot}, \eqref{xi_dot}, and \eqref{eta_dot}. 

Now we show that the unforced outer subsystem $\dot{e}_3=f_3\left({e}_3, 0\right)$ is Globally Asymptotically Stable (GAS). In this direction consider $W_3 \coloneqq e_{\xi_3}^2/2+e_{\eta_3}^2/2$ as a candidate Lyapunov function corresponding to agent 3's formation errors. Taking the time-derivative of $W_3$ and substituting \eqref{control_third} into \eqref{xi_dot} and \eqref{eta_dot} gives:

\begin{equation}
\dot{W}_3 = -\frac{ 2 \left(\cosh \eta_3 - \cos\xi_3\right)}{\lVert p_{21}\rVert}(\kappa_3 e_{\xi_3}^2+\lambda_3 e_{\eta_3}^2),
\end{equation}
where $\cosh \eta_3 - \cos\xi_3> 0$. This is because $\cosh \eta_3 \geq 1$ and its minimum occurs at $\eta_3=0$, while $\cos \xi_3 \leq 1$ and its maximum occurs at $\xi_3= 0$. Note that these extremum conditions never occur at the same time for any position of agent 3 since the isoquant surface of $\eta_3 = 0$ is $Y_3-Z_3$ plane, while $\xi_3 = 0$ occurs on $X_3$ axis (See Fig. \ref{frame2d}) except for points on the line passing through agents 1 and 2. Since the surface and lines do not intersect each other, the above statement is true. As a result, $\dot{W}_3 < 0$ holds except at the equilibrium point $e_3 = [0,0]$, establishing GAS of the unforced outer subsystem $\dot{e}_3=f_3\left({e}_3, 0\right)$. Now let $\mathcal{B}({e_2}=0)$ be the basin of attraction of the equilibrium point $e_2 = 0$ of $g_2(e_2)$. Note that $\mathcal{B}({e_2}=0) = \mathbb{R}\backslash\{{-d_{21}^*}^2\}$ and $e_3$ in $E_3 \times \mathcal{B}({e_2}=0)$ is bounded, where $E_3 = (-\pi, \pi) \times \mathbb R$. Thus, \cite[Theorem 1]{Welde-aGAS} proves a-GAS of the cascade system in \eqref{cascade-three}.  

As we grow the formation error system by adding agent $l$ in the tetrahedral subgraph of $\mathcal{G}$ with neighboring agents $i$, $j$, $k$ where $i<j<k<l,\; l\geq 4$,  we obtain the following cascade system at each step:
% \vspace{-.1cm}
\begin{equation}
\begin{aligned}
& \dot{e}_l=f_l\left({e}_l, \Xi_{l-1}\right), \\
& \dot{\Xi}_{l-1}=g_{l-1}\left(\Xi_{l-1}\right),
\end{aligned}
% \vspace{-.1cm}
\label{cascade-l}
\end{equation}
where $e_l \coloneqq [e_{\xi_l},e_{\eta_l},e_{\varphi_l}]$ is the vector of stacked formation errors for agent $l$ and $\Xi_{l-1} \coloneqq [e_2,e_3,...,e_{l-1}]$ is the vector of stacked formation errors for agents $2,3,...,l-1$. Note that the exact expressions of $f_l$ and $g_{l-1}$ can be found by substituting control laws \eqref{control_second} and \eqref{control_third}, and \eqref{control_follower} into \eqref{distance_dot}, \eqref{xi_dot}, \eqref{eta_dot}, and \eqref{phi_dot}.
Now we argue that the unforced outer subsystem $\dot{e}_l = f_l\left({e}_l, \bold{0}_{\Xi_{l-1}}\right)$ is GAS where $\bold{0}_{\Xi_{l-1}}$ is vector of zeros with $3l-9$ entries corresponding to equilibrium point of the inner subsystem $\dot{\Xi}_{l-1}=g_{l-1}\left(\Xi_{l-1}\right)$ in \eqref{cascade-l}. Consider $W_l = e_{\xi_l}^2/2+e_{\eta_l}^2/2 + e_{\varphi_l}^2/2$ as a candidate Lyapunov function corresponding to agent $l$'s formation errors. Applying the time-derivative of $W_l$ and  substituting \eqref{control_third} into \eqref{xi_dot} and \eqref{eta_dot} gives: 
\begin{equation}
    \dot{W}_l = -\frac{ 2 \left(\cosh \eta_l - \cos\xi_l\right)}{\lVert p_{ij}\rVert}(\kappa_l e_{\xi_l}^2+\lambda_l e_{\eta_l}^2+\gamma_l\frac{e_{\varphi_l}^2}{\sin\xi_l}).
\end{equation}
One can show that $\dot{W}_l < 0$ except at the equilibrium point $e_l = [0,0,0]$ leading to GAS of $\dot{e}_l = f_l\left({e}_l, \bold{0}_{\Xi_{l-1}}\right)$.
Now let $\mathcal{B}({\Xi_{l-1}}= \bold{0_{\Xi_{l-1}}})$ be the basin of attraction of the equilibrium point ${\Xi_{l-1}}= \bold{0_{\Xi_{l-1}}}$ of $g_{l-1}({\Xi_{l-1}})$. Note that $e_l$ in $E_l \times \mathcal{B}({\Xi_{l-1}=\bold{0_{\Xi_{l-1}}}})$ is bounded, where $E_l = (-\pi,\pi)\times \mathbb{R}\times(-2\pi,2\pi)$ and $\mathcal{B}({\Xi_{l-1}}= \bold{0_{\Xi_{l-1}}}) = E_{l-1} \times ... \times E_3 \times \mathcal{B}(e_2 = 0)$. Thus, \cite[Theorem 1]{Welde-aGAS} establishes a-GAS of the cascade system in \eqref{cascade-l}. Note that a-GAS property of the cascade system ensures boundedness of formation errors of the agents. This means $\eta_3,\eta_4 ..., \eta_l$ are bounded meaning that every agent will never get collocated on its neighboring agents and avoid collision with them ($\lVert p_{ij}\rVert \neq 0,\; \forall (j,i) \in \mathcal{E} $) since $\eta_l = \pm \infty$ occurs at the foci of bispherical coordinates given in \eqref{bi2cart}.   
\end{proof}

\section{Simulation Results}
Consider six agents that are distributed at random positions (leader is in origo.) in a 3D workspace at $t=0$. The objective is to form a unit octahedron until $t = 10$ seconds and then double up its scale until $t =20$ seconds under control law \eqref{controllaws}. The edge set of the directed sensing graph (obeying Assumption \ref{assum:G}) is considered as $\mathcal{E}=$ $\{(2,1),(3,1),(3,2),(4,1),(4,2),(4,3),(5,2),(5,3),(5,4)$, $(6,3),(6,4),(6,5)\}$. The sensing graph among the agents is depicted in \ref{sensing}. The desired shape is characterized by setting all of the desired lengths in the edge set equal to 1 except for $d_{32}^*=d_{64}^*= \frac{\sqrt{2}}{2}$. The desired signed volumes are assumed to be $V_{1234}^*=V_{2345}^*=-V_{3456}^*=\frac{\sqrt{2}}{12}$. To realize the formation scaling the desired distance between agents 1 and 2 is initially set to $d_{21}^* = 1$ until $t = 10$ and then it is changed to $d_{21}^* = 2$. 
The control gains in \eqref{control_second},\eqref{control_third}, and \eqref{control_follower} are all set to 2. The trajectory of each agent during $t \in [0,10]$, where the desired formation scale is equal to 1 is shown in Fig. \ref{initial-traj}. Fig. \ref{scaling-traj} shows the final scaled-up shape as well as the trajectories of the agents during $t \in [10,20]$.

The evolution of squared distance error and bispherical formation errors, defined in \eqref{error_distance}, \eqref{error_xi}, and \eqref{error_phi}, are shown in Fig. \ref{error_fig}. The results show that the agents successfully converge to the desired formation after approximately 10 units of time despite their random initial positions. Moreover, after the step shift in agent 2's desired distance at $t=10s$, which induces a sudden increase in the follower's formation errors (see Fig. \ref{error_fig}), the formation eventually scales up as soon as the follower's formation errors converge back to zero. 

\begin{figure}[!tbp]
    \centering
    \begin{subfigure}[b]{.69\linewidth} % Adjust width to fit content
        \centering
        \includegraphics[width=\linewidth]{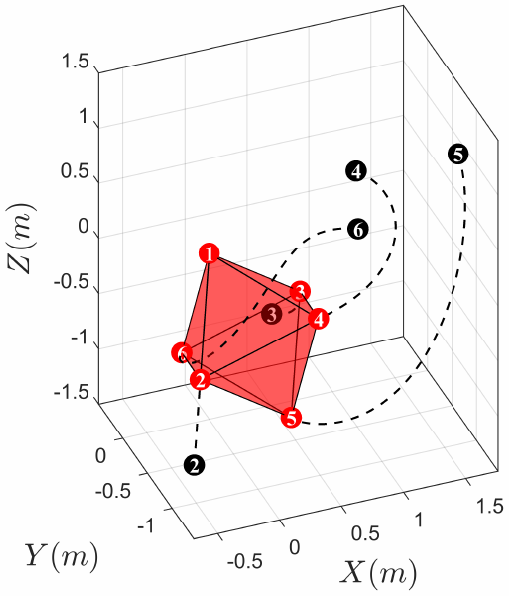}
        \caption{}
        \label{initial-traj}
    \end{subfigure}

    \begin{subfigure}[b]{.69\linewidth} % Adjust width to fit content
        \centering
        \includegraphics[width=\linewidth]{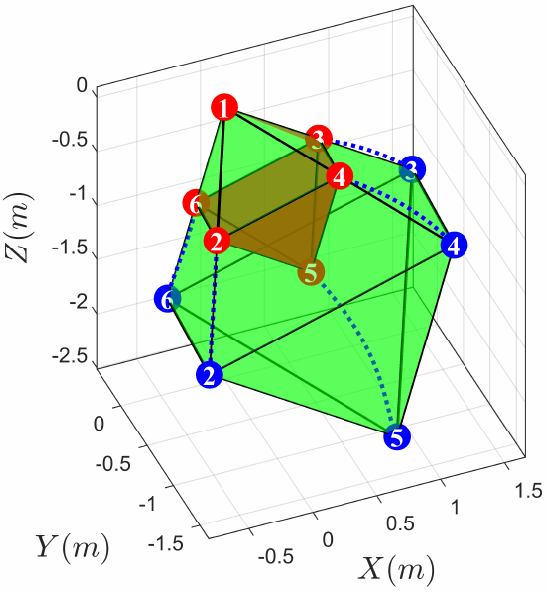}
        \caption{}
        \label{scaling-traj}
    \end{subfigure}
    \caption{(a) Trajectory of agents until $t=10$.  (b) Scaling simulation: agents scaling up after $t=10$ because of a change in $d_{21}^*$.}
    \label{traj}
\end{figure}

\begin{figure}[!tbp]
	\centering
	\begin{subfigure}[t]{.85\linewidth} % Adjust the width as necessary
		\centering
		\includegraphics[width=\linewidth]{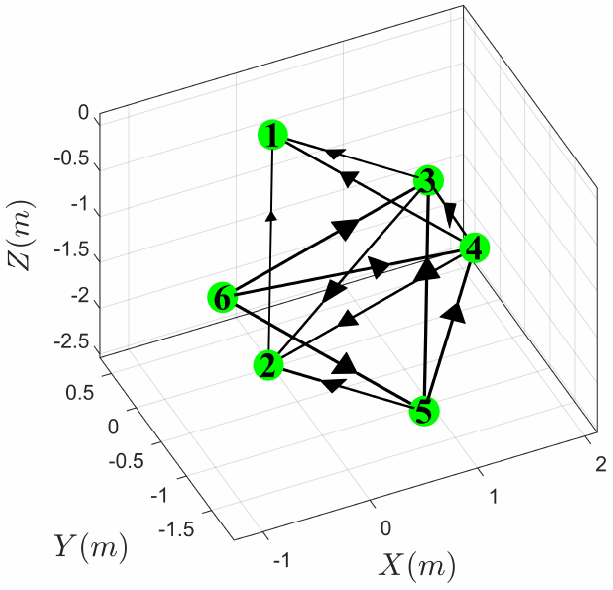}
		\caption{}
		\label{sensing}
	\end{subfigure}
	
	% Adding some vertical space between the figures
	\vspace{1cm} % Adjust the space as necessary
	
	\begin{subfigure}[t]{.85\linewidth} % Adjust the width as necessary
		\centering
		\includegraphics[width=\linewidth]{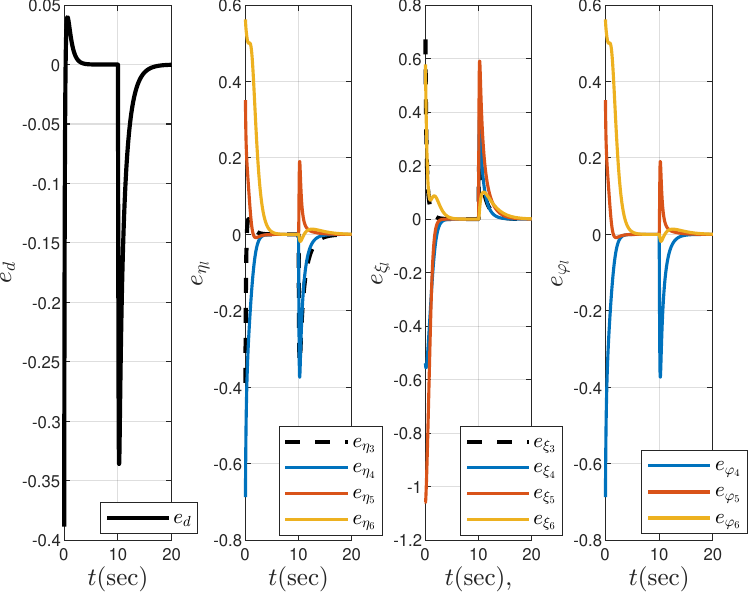}
		\caption{}
		\label{error_fig}
	\end{subfigure}
	\caption{(a) Sensing graph among agents. (b) Formation errors vs time.}
	\label{traj}
	\vspace{-.6cm} % You may need to adjust this spacing based on your document layout
\end{figure}
 \section{Conclusion}
We introduced a novel 3D formation control scheme, with (almost) global convergence to the desired shape under acyclic triangulated directed sensing topologies. We utilized orthogonal bispherical coordinates to uniquely characterize the desired formation shape and effectively avoid undesired equilibria imposed by the controller design. Our proposed control scheme also allows for formation scaling through adjustments of the distance between agents 1 and 2. Applying the stability theory of cascade-connected systems, we established that the proposed decentralized controllers make the closed-loop system almost globally asymptotically stable. Furthermore, we reasoned that the proposed control scheme can be readily implemented in arbitrarily oriented local coordinate frames of the (follower) agents using onboard vision sensors.
Future work will be devoted to extending these findings to agents with more complex dynamics and addressing the formation maneuvering problem.

\bibliographystyle{ieeetr}
\bibliography{Refs}

%%%%%%%%%%%%%%%%%%%%%%%%%%%%%%%%%%%%%%%%%%%%%%%%%%%%%%%%%%%%%%%%%%%%%%%%%%%%%%%%
\appendices

\section{Error Dynamics}\label{error_dynamics}
\vspace{-.1cm}
Taking the time derivative of \eqref{error_distance} and substituting \eqref{control_second} gives:
\begin{equation}
\dot{e}_d=-2 \kappa_2\left\|p_{21}\right\|^2 e_d=-2 \kappa_2 e_d(e_d+{d_{21}^*}^2).
\label{distance_dot}
\end{equation}
Taking the time derivative of $e_\xi$ in \eqref{error_xi} and using \eqref{xi_l}, while following a similar approach as in \cite{buckley2021infinitesimal}, one can derive the error dynamics of the first bispherical coordinate for agent $l\; \geq 3$ as follows:

\begin{equation}
\dot{e}_{\xi_l}=\dot{\xi_l}=\dot{\theta}_{ilj}=N_{j li} \dot{p}_i-\left(N_{j l i}+N_{i lj}\right) \dot{p}_l+N_{i l j} \dot{p}_j,
\label{xi_dot}
\end{equation}

where 
\begin{equation}
\begin{gathered}
   N_{jli}=\frac{-\hat{v}_{lj}^{\top} P_{\hat{v}_{l i}}}{\lVert{p_{l i}}\rVert \sin \xi_l}, \quad N_{ilj}=\frac{-\hat{v}_{li}^{\top} P_{\hat{v}_{l j}}}{\lVert{p_{l j}}\rVert \sin \xi_l},\\ P_{\hat{v}_{l j}} = I_3-\hat{v}_{l j}\hat{v}_{l j}^\top,\quad P_{\hat{v}_{l i}} = I_3-\hat{v}_{l i}\hat{v}_{l i}^\top. 
   \end{gathered}
\end{equation}
 Rewriting elements of \eqref{xi_dot} in $\{C_l\}$ using \eqref{bi2cart}, one can verify the followings after some algebraic simplifications:

\begin{equation}
\begin{aligned}
    &N_{j li}=\frac{1}{2a_l}\left[\begin{array}{c}-e^{-\eta_l} \sin \xi_l\\\cos \varphi_l\left(1-e^{-\eta_l} \cos \xi_l\right)\\ \sin \varphi_l\left(1-e^{- \eta_l} \cos \xi_l\right)\end{array}\right]^\top,\\ &N_{ilj}=\frac{1}{2a_l}\left[\begin{array}{c}e^{\eta_l} \sin \xi_l\\\cos \varphi_l\left(1-e^{\eta_l} \cos \xi_l\right)\\\sin \varphi_l\left(1-e^{\eta_l} \cos \xi_l\right)\end{array}\right]^\top,\\ -&\left(N_{jli}+N_{ilj}\right) =
\frac{\left(\cosh\eta_l-\cos\xi_l\right)}{a_l}\hat{\xi}_l^\top.
\label{xi_dot_terms}
\end{aligned}
\end{equation}
Moreover, taking the time derivative of $e_{\eta_l}$ in \eqref{error_xi} and using \eqref{eta} yields \cite{mehdifar20222}:
\begin{equation}
\dot{e}_{\eta_l}=\dot{\eta_l} =\frac{p_{li}^\top\dot{p}_{li}}{\lVert{p_{li}}\rVert^2}-\frac{p_{lj}^\top\dot{p}_{lj}}{\lVert{p_{lj}}\rVert^2}= M_i \dot{p}_i + M_j\dot{p}_j + M_l\dot{p}_l.
\label{eta_dot}
\end{equation}
Once again, one can verify the followings by writing elements of \eqref{eta_dot} in $\{C_l\}$ using \eqref{bi2cart}:
\begin{equation}
\begin{gathered}
M_i=\frac{-1}{2 a_l}\left[\begin{array}{c}
1-e^{-\eta_l} \cos \xi_l \\
e^{-\eta_l} \sin \xi_l \cos \varphi_l \\
e^{-\eta_l} \sin \xi_l \sin \varphi_l
\end{array}\right]^\top,\\ M_j = -\frac{1}{2 a_l}\left[\begin{array}{c}
1-e^{\eta_l} \cos \xi_l \\
-e^{\eta_l} \sin \xi_l \cos \varphi_l \\
-e^{\eta_l} \sin \xi_l \sin \varphi_l
\end{array}\right]^{\top},\\M_l 
= \frac{\left(\cosh\eta_l-\cos\xi_l\right)}{a_l}\hat{\eta}_l^\top. \end{gathered}
\end{equation}
Applying the time derivative to $e_{\varphi_l}$ in \eqref{error_phi} and using \eqref{alpha} and \eqref{bi2cart} gives:
\begin{equation}
    \dot{e}_{\varphi_l}=\dot{\varphi_l} = L_i \dot{p}_i + L_j\dot{p}_j + L_k\dot{p}_k +L_l\dot{p}_l,
\label{phi_dot} 
\end{equation}
where:
\begin{align}
          L_i&= -\frac{{\hat{\varphi}_l}^\top\left(e^{-\eta_l}-\cos \xi_l\right)}{2 a_l \sin \xi_l}
     +\frac{\hat{Z}_l^\top\left(e^{-\eta_k}-\cos \xi_k\right)}{2 a_l \sin \xi_k}, 
\notag\\L_j&=-\frac{{\hat{\varphi}_l}^\top\left(e^{\eta_l}-\cos \xi_l\right)}{2 a_l \sin \xi_l}+\frac{\hat{Z}_l^{\top}\left(e^{\eta_k}-\cos \xi_k\right)}{2 a_l \sin \xi_k}, 
\\L_k &=-\frac{\hat{Z}_l^\top\left(\cosh \eta_k-\cos \xi_k\right)}{a_l \sin \xi_k}  ,\; L_l =\frac{{\hat{\varphi}_l}^\top\left(\cosh \eta_l-\cos \xi_l\right)}{a_l \sin \xi_l}.  \notag
\end{align}

\end{document}